\documentclass{amsart}

\usepackage{fullpage}
\usepackage{cite}
\usepackage{graphicx}
\usepackage{subfig}
\usepackage[sans]{dsfont}

\usepackage[bookmarks,colorlinks,breaklinks]{hyperref}  % PDF
\hypersetup{colorlinks=true, linkcolor=red, filecolor=blue, urlcolor=blue, citecolor=green}

\newtheorem{theorem}{Theorem}[section]
\newtheorem{proposition}{Proposition}[section]
\newtheorem{lemma}{Lemma}[section]
\newtheorem{remark}{Remark}[section]
\newtheorem{definition}{Definition}[section]
\newtheorem{hyp}{Hypothesis}[section]

\newcommand{\R}{\mathbb{R}}
\newcommand{\RR}{\mathcal{R}}
\newcommand{\C}{\mathbb{C}}

\newcommand{\T}{\mathcal{T}}
\newcommand{\n}{\noindent}
\newcommand{\m}{\medskip}
\newcommand{\D}{\displaystyle}

\newcommand{\diag}{\mathop{\mathrm{diag}}}

\newcommand{\NGO}{\mathcal{N}}
\newcommand{\bNGO}{\bar{\mathcal{N}}}
\newcommand{\bX}{\bar{X}}
\newcommand{\bY}{\bar{Y}}
\newcommand{\bZ}{\bar{Z}}
\newcommand{\one}{\mathds{1}}
\newcommand{\hp}{\mathop{\circ}}

\newcommand{\bA}{\bar{A}}
\newcommand{\bB}{\bar{B}}
\newcommand{\DD}{\mathcal{D}}

\newcommand{\V}{\mathcal{V}}
\newcommand{\E}{\mathcal{E}}
\newcommand{\M}{\mathcal{M}}
\newcommand{\MM}{\mathfrak{M}}
\newcommand{\CC}{\mathcal{C}}
\newcommand{\tT}{\tilde{T}}

\newcommand{\bxi}{\bar{\xi}}

\newcommand{\bdX}{\mathbf{X}}
\newcommand{\rd}{\mathrm{d}}

\begin{document}

\title[Multi-group models for vector-borne diseases]{On the dynamics of a class of  multi-group  models for vector-borne diseases
}

\author{ Aberrahman Iggidr}
\author  {Gauthier Sallet}
\address[A. Iggidr and G. Sallet]{Inria, Universit\'e de Lorraine, CNRS. Institut Elie Cartan de Lorraine, UMR 7502. ISGMP Bat. A, Ile du Saulcy, 57045 Metz Cedex 01, France. }

\author{Max O. Souza}
\address[M.O. Souza]{Departamento de Matem\'atica Aplicada, Universidade Federal Fluminense,
R. M\'ario Santos Braga, s/n
Niter\'oi - RJ, 24020-140, Brasil.}

\thanks{This  work was supported by the CNPq Pronex/Dengue network grant \# 550030/2010-7 and by  COFECUB/CAPES project 709-2010. MOS was partially supported by CNPq grant \# 308113/2012-8.}



\begin{abstract}
The resurgence of vector-borne diseases is an increasing public health concern, and there is a need for   a better understanding of their dynamics.  For a number of  diseases, e.g. dengue and  
chikungunya, this resurgence occurs  mostly in urban environments,  which are naturally very  heterogeneous, particularly due to population circulation.  In this scenario,  there is an increasing interest in both multi-patch and multi-group models for such diseases.   In this work, we study the dynamics of a vector borne disease within  a class of  multi-group models that extends the classical Bailey-Dietz model. This class includes many of the proposed models in the literature, and it can accommodate various functional forms of the infection force.  For such models, the vector-host/host-vector  contact   network  topology gives rise  to a bipartite graph which has  different  properties from  the ones usually found in directly transmitted diseases.  Under the assumption that the contact network is strongly connected, we can define the basic reproductive number $\RR_0$ and show that  this system has only two equilibria: the so called disease free equilibrium (DFE); and a unique interior equilibrium---usually termed the endemic equilibrium (EE)---that exists if, and only if, $\RR_0>1$.  We  also show that, if $\RR_0\leq1$, then the DFE equilibrium is globally asymptotically stable, while when $\RR_0>1$, we have that   the  EE is globally asymptotically stable.
\end{abstract}

\maketitle

%\linenumbers

\section{Introduction}

\subsection{Background}

The global resurgence of vector-borne diseases is  a growing concern for public health officers in  many countries \cite{Gubler}. Diseases like dengue and chikungunya continue   to spread all over the world, hand in hand with the spread of their associated vectors; cf. \cite{Powers:2000kl}. Thus, in the United States the {\it Aedes albopictus}, the tiger mosquito,  is fixating very rapidly , while in Europe  {\it Ae. albopictus} is  also spreading at a fast rate---cf. \cite{Lambrechts:2010kh}. The result of this fixation  is already evident: Italy and the South of France have already had documented cases of  chikungunya \cite{cdc:2014}, and there is a growing number of dengue cases detected in the US \cite{anez:rios:2013}. Furthermore, dengue is now the leading cause in US  of acute febrile state of 
travelers returning from Asian, South American and Caribbean countries \cite{cdc:2010}.  In the particular case of dengue, the main vector, \textit{Ae. Aegypti}, is  anthropophilic, and it lives only on urban or semi-urban areas. It is also a very sedentary mosquito: it will usually fly no more than about five hundred meters from its birth place, unless in extreme adverse conditions. These observations suggest that one  should not expect that dengue will spread through the diffusion of the vector.

Indeed a  number of such resurgent diseases occur in highly urban areas and are transmitted by vectors that do not disperse very far compared to other species---cf. \cite{Honorio:etal:2003} and references therein. On the other hand,   in the case of a urban area  with an efficient  transportation system,  movements from one location to another are fast. Then, for a given individual, disease transmission will mostly likely happen either at  its home region or   at  its usual  destination location. In this scenario, susceptible individuals can become infected in areas that are geographically apart from their residence area, and  infected individuals can travel quite long distances and be able to infect vectors in very distinct areas were they themselves infected. Since the  disease dynamics is likely to be  largely dependent on whether one has a homogeneous or a  heterogeneous population, with heterogeneity 
favoring the establishment of epidemics---cf. \cite{Hasibeder,Levin,Smith}---this suggests that  in areas with significant population movement, the epidemiological dynamics can be strongly influenced  by the circulation of human hosts. The link between host circulation and the disease dynamics   seems to be first pointed out by \cite{Adams,Cosner2009,Stoddard} in slightly different frameworks. In any case, circulation naturally segregates host and vector by their registered and current location, and   it is then natural to consider the so-called meta-population models as candidates for modeling the disease dynamics. Such  meta-population  models can be either of  multi-patch or of multi- group type. In some regimes, the latter can arise as  a limit model of the former---eg.  in the case  of fast sojourn times; cf. \cite{Adams}. 

The previous discussion suggests that   the use of  multi-group models might become a valuable  modeling tool for understanding the disease dynamics in urban settings, and indeed there is a growing interest in the literature on these models. See \cite{Smith:etal:2014} for a recent review on such models,  and for a discussion on their importance in the epidemiological modeling, and \cite{mpolya:etal:2013} for a study  in a star network.  In addition, see also \cite{cambodia} for empirical studies on the impact of human movement on the disease dynamics and \cite{Alvimetal2013} for complementary views to \cite{Adams,Cosner2009}.  For a theoretical review on multi-group models, see \cite{MR1993355}.

The overall interest in these epidemic models have, in turn,   raised a natural  interest  in  understanding   their qualitative dynamical properties. This has fostered a considerable literature addressing this problem, and  which we now briefly review.

\subsection{Disease dynamics}

From the point of view of epidemiological mathematical modeling, the first natural question about any  disease-dynamics model is  what are its stability features as a function of the basic reproduction number, $\RR_0$. Following \cite{Shuai:Driessche:2013},  we say than an epidemic model has the \textit{sharp $\RR_0$ property} if the following holds:  when $\RR_0\leq1$, the only feasible equilibrium is the so-called disease free equilibrium (DFE), and it is globally asymptotically stable (GAS);  when $\RR_0>1$, there is a single interior equilibrium, the so-called endemic equilibrium (EE), which is then GAS.

The literature on mathematical epidemiology and the study of Sharp $\RR_0$ property is long and large,  particularly for directly transmitted diseases, but it is considerably smaller for vector-borne diseases. The development of the models for indirectly transmitted diseases can be traced back to Ross malaria model as discussed in \cite{Ross1911}---see also the recent review in \cite{Smith1} and the classical monographs \cite{Bailey75,Dietz75}.  Nevertheless, the bulk of the theory in the literature is leaned towards directly transmitted diseases and uniform populations---see \cite{AndMay91,Diekmann:2000} for instance. For vector-borne diseases, a very natural model is the coupling of a SIR model for the humans with a SI model. This model is reasonable for mosquito borne diseases, since they do not have a well developed immunological system, while most of the arboviruses confer lifetime stability. This model seems to be first suggested in \cite{Bailey75,Dietz75} and it is now known  as the Bailey-Dietz model. 
The global dynamics of this model was first studied in   \cite{9656647} using    a Lyapunov function argument for   the stability of the DFE, while  the Poincar\'e-Bendixson property for 3-D competitive systems  is used to show  the stability of the EE;  see also \cite{MR2653590,MR2559889} for later similar studies. A global stability analysis using only Lyapunov functions has been obtained only recently---\cite{Souza2014}. See also \cite{0533.92023,Lietal:1999} for various results on global stability of epidemiological models.

In the framework of multi-group epidemic models for directly transmitted diseases, the first paper was probably by Rushton and Mauser \cite{MR0069470}, but seminal results are  in  Lajmanovich and  Yorke  \cite{LajYo76} and in the book of Hethcote and Yorke  \cite{HethYor84}; but see also \cite{Nold:1980}. Stability results can be found in Thieme \cite{0533.92023,MR87c:92046};  see  also chapter 23 of \cite{MR1993355}.   Global stability of multi-group SIR model  is due to \cite{gls:2006}  by using a combinatorial argument arising from graph theory; see also \cite{gls:2008} for  a more extensive presentation of their method.   For indirectly transmitted diseases, the first global stability result seems to be due to   \cite{Hasibeder}, who observed  that a monotone dynamics  argument of \cite{LajYo76} was also applicable to a SI-SI multi-group model. More recently, general global stability results were obtained by \cite{Shuai:Driessche:2013}; see also \cite{Guo:etal:2012} for results on multi-stage models. None of these results, however, cover the case of vector-borne diseases, since vector and host populations might follow different dynamics.   Additional references in meta-population  models for vector-borne diseases, but without studying the sharp $\RR_0$ property are \cite{Honorio:2009uq,MBS6900} for models with  heterogeneous populations and \cite{Xiao:Zou:2014}  for a  numerical study of a multi-patch  model with spatial heterogeneities. 

For higher dimensional systems, global stability of endemic equilibrium is usually done by finding an appropriate Lyapunov function---\cite{Hasibeder} being a notable exception. The use of Lyapunov functions to study the global dynamics of ecological and epidemiological models can be traced at least to the works in the late seventies of  Goh  \cite{goh:1977,goh:1978,goh:1979,goh:1980}, Harrison \cite{Harrison:1979,Harrison:1979b} and Hsu \cite{Hsu:1978}.  Since then, it has been successfully used in many studies, and even rediscovered \cite{Freedman:So:1985,Koro:2001,KoroMMB04,Korobeinikov:Maini:2004,Koro:2006,MR2434863}. Recent applications of Lyapunov functions in  epidemic and ecological  models with meta-populations  include \cite{Iggidr:etal:2006,Koro:2009,Yu:etal:2009,Li:etal:2010,Li:Shuai:2010,Ji:etal:2011,Kuniya:2011,Souza:Zubelli:2011,Sun:Shi:2011,Guo:etal:2012,Huang:etal:2012,Shuai:Diressche:2012,Magal:McCluskey:2013,Wang:2014}. See also the recent surveys on the construction and use of Lyapunov functions in models of population dynamics by \cite{Hsu:2005,MR2434863}. Additionally, there  is also recent work  aiming to obtain similar results for multi-group models, but  without recurring to graph theoretic arguments \cite{Li:etal:2012,muroya:etal:2013}.   Shuai and van den Driessche~\cite{Shuai:Driessche:2013} discuss  two systematic approaches (graph-theoretic and matrix-theoretic) to guide the construction of Lyapunov functions.  For results  towards  infinite dimensional problems, see \cite{Thieme:2011}.

In this work, we show that the sharp $\RR_0$ property holds for a very natural multi-group extension of the  Bailey-Dietz model---that has been used  to model, \textit{inter alia}, the dynamics of dengue \cite{Nishiura:2006}. This extension also accommodates a large number  of choices for the  modeling of the infection-force, including the most popular ones---see \S\ref{sec:prelims} for an additional discussion on this issue. A special case within the class of models discussed here was  studied in \cite{ding:etal:2012} which, however, present an incorrect proof of the global stability of the endemic equilibrium\footnote{The matrix whose kernel should yield the coefficients for Lyapunov function is actually not singular for $n>2$. For $n=2$, a careful checking shows that the claimed cancellation properties  do not hold.}. This work can also seen as an  extension of the multi-group framework for direct-transmitted diseases  in \cite{gls:2006,gls:2008}.

\subsection{Outline}

In Section~\ref{sec:prelims} we introduce the relevant class of multi-group models and identify the relevant network structure, which is a bipartite graph, that we term the host/vector network. This  bipartite graph can be reducible, even when the group network is strongly connected. This is markedly different from directly transmitted diseases. On the assumption that  the host/vector network is strongly connected, we can meaningfully define an $\RR_0$.  For the models discussed here, the existence and uniqueness of the Endemic Equilibrium (EE) , when,  $\RR_0>1$ is not obvious from the governing equations,  and these issues are  tackled in Section~\ref{sec:eqstab}, where the local stability is also established.   We then study the global dynamics in section~\ref{sec:global}:   when $\RR_0\leq1$, we show that  the  disease free equilibrium (DFE) is globally asymptotically stable .  We then address the global stability of the EE and,  we then show that it is  globally asymptotically stable when $\RR_0>1$ using a "vectorial" extension of the  Lyapunov function used in \cite{Souza2014} together with an extension of the graph-theoretical approach developed in \cite{gls:2006,gls:2008}. A discussion of the results is given  in Section~\ref{sec:concl}.

\section{A class of multi-group models for vector-borne diseases}

\label{sec:prelims}

In the following, we provide the basic set up for a class of multi-group models for indirectly transmitted diseases. These models are built upon the classical single-patch/group model by \cite{Bailey75,Dietz75}, and include some of the models studied in \cite{Adams,Cosner2009} and the models studied in \cite{Alvimetal2013}. 

\subsection{ The basic model} %l and the multigroup extension}
\label{ssec:onep}

We consider the classical Bailey-Dietz   model:
\begin{equation}\label{1patch}
\left \{
\begin{array}{ll}
\dot  S_h=&\Lambda_h- \beta_1 \, \dfrac{S_h \, I_v}{N_h}  -\mu_h \, S_h \\
 \dot I_h=& \beta_1 \, \dfrac{S_h \, I_v}{N_h} - \gamma_h \, I_h- \mu_h \, I_h \\
 \dot R_h=& \gamma_h \, I_h-\mu_h \, R_h \\
 \dot S_v =& \Lambda_v - \beta_2 \, \dfrac{S_v \, I_h}{N_h} -\mu_v \, S_v\\
 \dot I_v= & \beta_2 \, \dfrac{S_v \, I_h}{N_h} -\mu_v \, I_v,
\end{array}
\right.
\end{equation}
where $S_h$, $I$, $R$ denote, as usual, the class of susceptible, infections and removed, respectively. The superscripts $h$ and $v$ indicate that the quantity refers to the host or to the vector. Also, $N_h=S_h+I_h+R_h$ and $N_v=S_v+I_v$ are the total host and vector, respectively, populations. Although they are not necessarily constant, they are taken as so in many applications.

The constant $\beta_1$ is a composite  biological constant that embodies all the biological processes relating to transmission from mosquito to man, from the biting rate of the mosquitoes through the probability to develop and infection after a bite. Analogously $\beta_2$ captures the effect of transmission from man to mosquito. The constant $\mu_h$ is the per capita human mortality, $\gamma_h$ denotes the per capita rates at which infectious individual recover and are permanently immune. The parameter $\Lambda_v$ is the constant recruitment of mosquitoes and $\mu_v$ is the per capita vector mortality.

Let
\[
\mathbf N = \dfrac{\Lambda_h}{\mu_h}  \textbf{ and }  \mathbf V=\dfrac{\Lambda_v}{\mu_v}.
\]

Using the techniques in \cite{VddWat02}, it is straightforward to see that the reproduction number of (\ref{1patch}) is

\[\mathcal R_0^2=    \dfrac{ \beta_1\,  \beta_2\,\mathbf V}{\mu_v\, (\mu_h+\gamma_h)\,  \mathbf N \,}=  \dfrac{ \beta_1\,  \beta_2\,\mathbf m}{\mu_v\, (\mu_h+\gamma_h)\,   \,}\]

with $\mathbf m= \dfrac{\mathbf V}{\mathbf N}$, the classical vectorial density. The basic reproduction ratio  $\mathcal R_0$ is the same than for a classical Ross's model \cite{AndMay91,MBS6900,Bailey75,Ross1911}.

As for Ross 's model we will use the prevalences, i.e.,  defining $x_1=\dfrac{S_h}{N}$, $x_2=\dfrac{I_h}{N}$, $x_3=\dfrac{R_h}{N}$ and $y_1=\dfrac{S_v}{V}$, $y_2=\dfrac{I_v}{V}$. Then, two equilibria are possible : the disease free equilibrium $( \mathbf 1, \mathbf 0, \mathbf 0,\mathbf 1,\mathbf 0)$ and, when $\mathcal R_0 >1$,  
a positive endemic equilibrium  $(\bar x_1, \bar x_2, \bar x_3, \bar y_1, \bar y_2)$.

The global stability of \eqref{1patch} was originally studied by \cite{9656647}, who showed that the  endemic equilibrium is globally asymptotically stable when $\mathcal R_0>1$, and that the disease-free is the global attractor when $\mathcal R_0\leq 1$. using the so-called Poincar\'{e}-Bendixson theorem for competitive systems---cf. \cite{0821.34003}.  More recently, \cite{Souza2014} has obtained a proof using only Lyapunov functions

\subsection{A class of multi-group models for vector-borne diseases}

We consider that  both  host and vector populations are divided in  $n$ groups, where  each group  $i$ has a host population of $N_{h,i}$ and a vector
population of $N_{v,i}$. At each node $i$, we assume  a generalized form of \eqref{1patch} by allowing that the susceptible of group $i$ to  have contact of mosquitoes of group $j=1,\ldots,n$. This is specified by an infection term for the host $\T_h$,  of the form
\[
\T_{h,i}=S_{h,i}\sum_{j=1}^nL_{i,j}(N_h,N_v)I_{v,j}.
\]
Analogously,  we allow susceptible mosquitoes of each group  $i$ to have contact with infected hosts group $j=1,\ldots,n$, with an infection force for the vectors, $\T_v$, of the form:
\[
\T_{v,i}=S_{v,i}\sum_{j=1}^nM_{i,j}(N_h,N_v)I_{h,j}.
\]
These assumptions then lead to the following multi-group epidemic model:
\begin{equation}\label{npatch2}
\left \{
\begin{array}{ll}
\dot  S_{h,i}=&\Lambda_{h,i}-    S_{h,i} \, \sum_{j=1}^n \, L_{i,j}(N_h,N_v)I_{v,j}     -\mu_{h,i} \, S_{h,i} \\[2mm]
 \dot I_{h,i}=& S_{h,i} \, \sum_{j=1}^n \, L_{i,j}(N_h,N_v)I_{v,j}  - \gamma_{h,i} \, I_{h,i}- \mu_{h,i} \, I_{h,i} \\[2mm]
 \dot R_{h,i}=& \gamma_{h,i} \, I_{h,i}-\mu_{h,i} \, R_{h,i} \\[2mm]
 \dot S_{v,i} =&\Lambda_{v,i} - S_{v,i} \, \sum_{j=1}^n \, M_{i,j}(N_h,N_v)I_{h,j}   -\mu_{v,i} \, S_{v,i}\\[2mm]
 \dot I_{v,i}= &  S_{v,i} \, \sum_{j=1}^n \, M_{i,j}(N_h,N_v)I_{h,j}   -\mu_{v,i} \, I_{v,i},
\end{array}
\right.
\end{equation}
where 
\[
N_h=(N_{h,i}),\text{ with } N_{h,i}=S_{h,i}+I_{h,i}+R_{h,i}
\text{ and }
N_v=(N_{v,i}),\text{ with } N_{v,i}=S_{v,i}+I_{v,i}.
\]
The functions $L_{i,j},M_{i,j}:\R^n\oplus\R^n\to\R$ are assumed to  be smooth and positive when $N_h,N_v$ have positive entries. These are mild assumptions, and they can accommodate a variety of functional forms for the infections force---see \cite{Wonham} for a discussion on the different conclusions implied by different assumptions on the infection force; see also  \cite{Alvimetal2013} for a discussion on the different transmission force related to dengue. These functions also encode the  cross-infection information  among all the groups, which will depend on the modeling assumptions that led to the multi-group structure. 

\begin{remark}
Similar models have been considered in the literature. See \cite{Cosner2009} for a multi-group SIS-SI model and \cite{Adams} for a  multi-group SEIR-SEI model,  obtained as the fast sojourn limit of a more general model. 
\end{remark}

\begin{remark}
While model \eqref{npatch2} can be easily modified to include disease induced death, the analysis carried out in the sequel cannot be extended to such models, except in the case of constant population. However, for diseases as dengue or chikungunya, this is  not a very restricting assumption,  as their morbidity is, generally,  not high. Dengue  can be an exception to that, if there are two epidemics in a row with an intermediate time spacing. In this case, enhanced immunological reaction can cause the so-called severe dengue fever, previously known as haemorraghic dengue, which can be highly fatal if not treated appropriately \cite{who:dengue,Gubl98}.
\end{remark}

We can rewrite \eqref{npatch2} as

\begin{equation}\label{npatch3}
\left \{
\begin{array}{ll}
\dot N_{h,i}=&  \Lambda_{h,i} -\mu_{h,i} \, N_{h,i} \\[2mm]
 \dot N_{v,i} =& \Lambda_{v,i}  -\mu_{v,i} \, N_{v,i}\\[2mm]
\dot  S_{h,i}=&\Lambda_{h,i}-  S_{h,i} \, \sum_{j=1}^n \, L_{i,j}(N_h,N_v)I_{v,j}    -\mu_{h,i} \, S_{h,i} \\[2mm]
 \dot I_{h,i}=&  S_{h,i} \, \sum_{j=1}^n \, L_{i,j}(N_h,N_v)I_{v,j}  - \gamma_h \, I_{h,i}- \mu_{h,i} \, I_{h,i} \\[2mm]
 \dot I_{v,i}= &  (N_{v,i}-I_{v,i}) \, \sum_{j=1}^n \, M_{i,j}(N_h,N_v)I_{h,j}   -\mu_{v,i} \, I_{v,i}.
\end{array}
\right.
\end{equation}
In what follows, we write $S_h=(S_{h,i})$, $i=1,\ldots,n$ and similarly for $I_h$ and $I_v$. Also, let
\[
\bar{N}_h=\left(\frac{\Lambda_{h,i}}{\mu_{h,i}}\right)
\text{ and }
\bar{N}_v= \left(\frac{\Lambda_{v,i}}{\mu_{v,i}}\right).
\]
Then, it is clear that, for \eqref{npatch3}, the set
\[\Omega= \{ (S_h,I_h,I_v,N_h, N_v) \in \R_+^{5n} \arrowvert   \;\;   0 \leq S_h+I_h \leq \bar{N}_h\,\;\;  0\leq  I_v \leq \bar{N}_v,\;\; 0\leq N_h\leq \bar{N}_h,\;\; 0\leq N_v\leq \bar{N}_v\}\]
is  a compact absorbing and positively invariant  set. 

Also, notice that  the system \eqref{npatch3} is of triangular form, and hence its stability analysis can be considerably simplified. There are a number of results that allow for such a simplification in the  study of global stability of systems of this kind\cite{0478.93044,Thieme:1992}. For the convenience of the reader, we recall the following result:
\begin{theorem}[Vidyasagar \cite{0478.93044}, Theorem 3.1]\label{Vid:theo}:

\noindent
Consider the following   $\mathcal C^1$ system :
 \begin{equation}\label{triang}
\left\{
\begin{array}{lr}
\dot x = f(x) &  x \in \R^n \; , y \in \R^m\\
\dot y = g(x,y)& \\
\text{ \rm with an equilibrium  point,  }(x^*,y^*), \text{ \rm    i.e.,} &\\
 f(x^*)=0   \mbox{\;} \text{\rm            and      }   g(x^*, y^*)=0 .&
\end{array}
\right.
\end{equation}

If  $x^*$ is  globally asymptotically stable (GAS)  in $\R^n$for the system  $\dot x= f(x)$,  and if  $y^*$ is GAS in   $\R^m$, for the system  $\dot y =g(x^*,y)$,  then   $(x^*,y^*)$ is (locally) asymptotically stable for  (\ref{triang}). Moreover, if all the trajectories of   (\ref{triang} ) are forward bounded,  then   $(x^*,y^*)$ is GAS for (\ref{triang}).
\end{theorem}

Since $(\bar{N}_h,\bar{N}_v)$ is a globally asymptotically stable equilibrium for the first two equations of \eqref{npatch3}, 
 we can use Theorem~\ref{Vid:theo} to reduce the study of the stability properties of \eqref{npatch3} to the study of the stability of 
\begin{equation}\label{npatch3b}
\left \{
\begin{array}{ll}
\dot  S_{h,i}=& \D \Lambda_{h,i}-  S_{h,i} \, \sum_{j=1}^n \, L_{i,j}(\bar{N}_h,\bar{N}_v)I_{v,j}    -\mu_{h,i} \, S_{h,i} \\
\\
 \dot I_{h,i}=& \D  S_{h,i} \, \sum_{j=1}^n \, L_{i,j}(\bar{N}_h,\bar{N}_v)I_{v,j}  - \gamma_{h,i} \, I_{h,i}- \mu_{h,i} \, I_{h,i} \\
 \\
 \dot I_{v,i}= &  \D 
(N_{v,i}-I_{v,i}) \, \sum_{j=1}^n \, M_{i,j}(\bar{N}_h,\bar{N}_v)
I_{h,j} -\mu_{v,i} \, I_{v,i}.
\end{array}
\right.
\end{equation}
In what follows, we shall denote by $\Lambda_h$, $\mu_h$ and $\gamma_h$ the vectors of 
$\R^n_+$ whose components are respectively $\Lambda_{h,i}$, $\mu_{h,i}$ and $\gamma_{h,i}$. We shall also write  $M=M(\bar{N}_h,\bar{N}_v)$ and $L=L(\bar{N}_h,\bar{N}_v)$.
%. 
System~\eqref{npatch3b} can then be written in the following vectorial notation: 
\begin{equation}\label{npatch4}
 \left\{
\begin{array}{ll}
\dot S_h=& \Lambda_h \,  - \mathrm{diag}(S_h)\, L\, I_v - \diag(\mu_h)  S_h\\
\\
\dot I_h=& \mathrm{diag}(S_h)\, L\, I_v - \diag(\mu_h+\gamma_h)  I_h\\
\\
\dot I_v=& \mathrm{diag}(\bar N_v-I_v) \, M \, I_h -\diag(\mu_v) I_v\,,
\end{array}
\right.
\end{equation}
where for  $\mathbf{v}\in\R^n$,  $\diag(\mathbf{v})$ denotes the $n\times n$ diagonal matrix whose main diagonal is $\mathbf{v}$.
\subsection{The Host-Vector contact network}

We shall need an assumption about the network topology  in system~\eqref{npatch4}.   For a matrix $A$, we write $\Gamma(A)$ for the associated graph. We begin with a definition
\vspace{3mm}
\begin{definition}[Host-Vector Contact Network]
	\label{def:hvcn}
Given nonnegative matrices $L$ and $M$, we write
\[
\mathcal{M}=\begin{pmatrix}
0&L\\
M&0
\end{pmatrix}.
\]
The graph associated to $\mathcal{M}$, $\Gamma(\mathcal{M})$, is denoted the host-vector contact network, or contact network for short.
\end{definition}
\vspace{3mm}

\begin{hyp}
\label{hyp:1}
The contact network is strongly connected, i.e.,  $\mathcal{M}$ is nonnegative and irreducible.
\end{hyp}
\vspace{3mm}
\begin{remark}
Notice that irreducibility of $L$ and $M$ are neither necessary nor sufficient for the irreducibility of $\mathcal{M}$. As an example, consider
\[
C=\begin{pmatrix}
0&1\\
1&0
\end{pmatrix}
\text{ and }
D=\begin{pmatrix}
1&0\\
1&1
\end{pmatrix};
\quad 
\mathcal{M}_1=
\begin{pmatrix}
0&C\\
C&0
\end{pmatrix}
\text{ and }
\mathcal{M}_2=
\begin{pmatrix}
0&D^t\\
D&0\\
\end{pmatrix}.
\]
Then $C$ is irreducible and $D$ is reducible. Nevertheless, $\mathcal{M}_1$ is reducible and $\mathcal{M}_2$ is irreducible.
\end{remark}

\vspace{3mm}

The irreducibility of $\mathcal{M}$ is associated to the strong connectivity of the corresponding directed bipartite graph. This is a consequence of the infection process, when considered  between hosts (or vectors) themselves, is a two step process. Thus, even when the circulation structure (the non-zero patterns of $L$ and $M$) is strongly connected, this is not necessarily the case for the host-vector contact structure of an indirectly transmitted disease, and this is a significant difference to directly transmitted ones.

In the following Proposition we shall give a useful characterization of the irreducibility of $\M$ that will be used later on:
\begin{proposition}
\label{prop:factors_irred}
 $\M$ is irreducible if, and only if, the following conditions are satisfied:
 \begin{enumerate}
 	\item Both  $LM$ and $ML$ are irreducible;
 	\item We have that $Lv,Mv\gg0$, for some  $v\gg0$ (and hence, for every $v\gg0$).
 \end{enumerate}
 Moreover, in this case, we also have that
 \[
 \rho(\M)^2=\rho(LM)=\rho(ML),
 \]
 and that both $LM$ and $ML$ have right and left positive eigenvectors associated to $\rho(\M)^2$.
\end{proposition}
\vspace{3mm}
\begin{proof}
Firstly, we compute
\[
\mathcal{M}^{2k}=\begin{pmatrix}
(LM)^k&0\\
0&(ML)^k
\end{pmatrix}
\quad\text{and}\quad
\mathcal{M}^{2k+1}=\begin{pmatrix}
0&L(ML)^k\\
M(LM)^k&0
\end{pmatrix},
\]
Assume  $\M$ is irreducible. Then there is some natural $n$ such that  
\[
(I+\M)^{2n}=\sum_{m=0}^{2n}\binom{2n}{m}\M^m=
\begin{pmatrix}
\sum_{k=0}^n\binom{2n}{2k}(LM)^k&L\sum_{k=0}^{n-1}\binom{2n}{2k+1}(ML)^k\\
M\sum_{k=0}^{n-1}\binom{2n}{2k+1}(LM)^k&\sum_{k=0}^{n}\binom{2n}{2k}(ML)^k
\end{pmatrix}
\gg0.
\]
Hence, we have that
\[
(I+ML)^n,(I+LM)^n\gg0,
\]
 and both $LM$ and $ML$ are irreducible as claimed. In addition, we have $L\sum_{k=0}^{n-1}\binom{2n}{2k+1}(ML)^k \gg0$. Thus $L$ applied to a column of $\sum_{k=0}^{n-1}\binom{2n}{2k+1}(ML)^k$ positive. The argument for $M$ is similar.
 
 Conversely, if both $LM$ and $ML$ are irreducible, then we 
have that  the main diagonals of $(I+\M)^{2n}$ are positive. The remaining blocks are also positive, since $L$ and $M$ are acting on positive matrices.
 
 Finally, let $(u,v)$ be a positive eigenvector associated to $\rho(\M)=\rho_\M$. The we necessarily have $Lu=\rho_\M v$ and $Mv=\rho_\M u$. Hence $LMv=\rho^2_\M v$, and similarly $MLu=\rho^2_\M u$. Furthermore $u$ and $v$ are positive right eigenvectors of $ML$ and $LM$, respectively, associated to $\rho_\M^2$. The argument for left eigenvectors is analogous.
\end{proof}

\begin{remark}
We observe that System~\eqref{npatch4} can be recast  as a special case of the multigroup SIR model treated in~\cite{gls:2006}, as follows: Replace $N_v-I_v$ by $S_v$ in the last equation of~\eqref{npatch4}. Include the redundant equation:
\[
\dot S_v= \Lambda_v \,  - \mathrm{diag}(S_v)\, M\, I_h - \diag(\mu_v)  S_v.
\]
Let $\mathbf{S}=(S_1,\ldots,S_{2n})^t$, $\mathbf{I}=(I_1,\ldots,I_{2n})^t$ and set $S_i=S_{h,i}$, $S_{n+i}=S_{v,i}$, $I_i=I_{h,i}$, $I_{n+i}=I_{v,i}$, for $i=1,\ldots,n$. Further, let $\M$ be as given in definition~\ref{def:hvcn} and let $\Lambda=(\Lambda_h\, \Lambda_v)^t$, $\mu=(\mu_h\, \mu_v)^t $, and $\gamma=(\gamma_h\, \mathbf{0})^t$. Then $(\mathbf{S}\,\mathbf{I})^t$ satisfies
\begin{align*}
\dot{\mathbf{S}}&=\Lambda-\diag(\mathbf{S})\M\mathbf{I}-\diag(\mu)\mathbf{S}\\
\dot{\mathbf{I}}&=\diag(\mathbf{S})\M\mathbf{I}-\diag(\mu+\gamma)\mathbf{I},
\end{align*}
for which the sharp threshold property holds.  Nevertheless,  we shall obtain this result  by considering equation \eqref{npatch4} directly, and using  a related but different  approach.  This can be seen as an extension to indirect transmitted  diseases of the framework introduced in \cite{gls:2006,gls:2008}. 

\end{remark}

\section{Equilibria and local stability}  

\label{sec:eqstab}

We will show that
for our vectorial disease with sub-populations structure,  System~\eqref{npatch2},  the results of \cite{MR87c:92046,MR1993355} are conserved.
Namely we
obtain that  the DFE is locally asymptotically stable, iff $\mathcal R_0 \leq
1$, and the existence and uniqueness of a strongly endemic
equilibrium when $\mathcal R_0 >1$. This equilibrium is always locally
asymptotically stable.  For global results, see Section~\ref{sec:global}.

Using the now standard techniques \cite{MR1057044,VddWat02}, we define the basic reproduction ratio as
	\[  \mathcal R_0 = \rho(\NGO),\quad \NGO= \left (
	\begin{array}{ll}
	0 & \diag(\mu_h+\gamma_h)^{-1} \,  \diag(\bar N_h) \,L \\
	\\
	\diag(\mu_v)^{-1}\,  \diag(\bar N_v) \, M &0
	\end{array} \right ).   \]

\begin{remark}
	\label{rem:about_ngo}
Since $\mu_v,\mu_h\gg0$, we have from Proposition~\ref{prop:factors_irred} that $\NGO$ is irreducible if, and only if, $\M$ is irreducible. In particular, if Hypothesis~\ref{hyp:1} holds then $\NGO$ is irreducible and we have that
	\[ \mathcal R_0^2= \rho \left ( \diag(\mu_h+\gamma_h)^{-1}\diag(\mu_v)^{-1} \,  \mathrm{diag}(\bar N_h) \,L   \,  \mathrm{diag}(\bar N_v) \, M \right ).
	\]
\end{remark}

\vspace{3mm}
\begin{theorem}
\label{thm:ee:exstab}
Assume  that hypothesis \ref{hyp:1} holds. Then system~\eqref{npatch4} (and hence system~\eqref{npatch2})  has a unique endemic equilibrium if, and only if, 
$\mathcal R_0>1$. Moreover this equilibrium is locally asymptotically stable with respect to System~\eqref{npatch4}.
\end{theorem}
\vspace{3mm}
\begin{proof}

We denote by $S_h^*$, $I_h^*$ and $I_v^*$ the expression
of an endemic equilibrium. Recall that the notation
$\mathds{1}$ refers to the vector of $\R^n_+$ whose components
are all equal to $1$. We have the following relation, defining an endemic equilibrium:

\begin{subequations}
 \begin{eqnarray}
\Lambda_h &=  \diag(\mu_h  + \, L\, I_v^*  ) \, S_h^*  \label{subeqnendem1}\\
 \diag (\mu_h+\gamma_h) \, I^*_h&=\diag(S_h^*)\, L\, I_v^* \label{subeqnendem2}\\
 \diag( \mu_v )\, I_v^*&=  \diag(\bar N_v-I_v^*) \,M\,I_h^* \label{subeqnendem3}
\end{eqnarray}
\end{subequations}

\n From \eqref{subeqnendem1} we obtain

\[ S_h^*=\diag(\mu_h + \, L\, I_v^*  )^{-1} \, \Lambda_h \]

\n Rewriting \eqref{subeqnendem3} as

\[ \diag(\mu_v) \, I_v^*=  \diag (M\,I_h^*)\,(\bar N_v-I_v^*) \]

Substituting for  $S_h^*$ in \eqref{subeqnendem2} we
obtain

\begin{subequations}
 \begin{eqnarray}
I_h^* &=  \diag(\mu_h+\gamma_h)^{-1} \diag(\mu_h + \, L\, I_v^*  )^{-1} \,\diag(L\,I_v^*) \,\Lambda_h\label{subeqnendem11}\\
   I^*_v&=  \diag(\mu_v  + \, M\, I_h^*  )^{-1} \, \diag(M\,I_h^*) \, \bar N_v \label{subeqnendem12}
\end{eqnarray}
\end{subequations}

\n Hence $(I_h^*,I_v^*)$ is a fixed point of the following
application

\[ F(x,y) =
\begin{bmatrix}
\diag(\mu_h + \gamma_h)^{-1} \diag(\mu_h  + \, L\, y  )^{-1} \,\diag(L \, y) \,\Lambda_h \\
 \\
 \diag(\mu_v  + \, M\, x  )^{-1} \, \diag(M\,x) \, \bar N_v
\end{bmatrix}\]
We will use a result of Hethcote and Thieme \cite{MR87c:92046}, which we recall for the convenience of the reader:

\begin{lemma}[\mbox{Theorem 2.1 in \cite{MR87c:92046}}]
	\label{thm:thieme}

Let $F(w)$ be a continuous, monotone non-decreasing, strictly
sublinear, bounded function which maps  the nonnegative orthant
$\R^n_+= [0, \infty)^n$ into itself. Let $F(0)=0$ and $F'(0)$
exist and be irreducible. Then $F(w)$ does not have a
nontrivial fixed point on the boundary of $\R^n_+$. Moreover,
$F(x)$ has a positive fixed point iff $\rho(F('0))>1$. If there
is a fixed point, then it is unique.

\end{lemma}

We have to check, for our function $F$ defined on $\R^n_+
\times \R^n_+$, the conditions of  Theorem~\ref{thm:thieme}.

It is immediate that $F$ is continuous, bounded and maps the
nonnegative orthant $\R^n_+ \times \R^n_+$ into itself.

 The function $F$ is monotone since the Jacobian of $F$ is

\[JF(x,y)=
\begin{bmatrix}
0 & A_1 \\
A_2 &0
\end{bmatrix}\]

With

 \[ A_1= \diag(\mu_h+\gamma_h)^{-1}\diag(\mu_h + \, L\, y  )^{-1} \, \diag(\Lambda_h) \left [I_n- \diag(\mu_h \ + \, L\, y  )^{-1} \,  \diag(L \, y)\right ] \,L.\]

 and

 \[ A_2= \diag( \bar N_v )\, \diag(\mu_v  + \, M\, x  )^{-1} \, \left [I_n - \diag(\mu_v + \, M\, x  )^{-1} \, \diag(M\,x) \right ] \, M.   \]

 \n
 Then $JF(x,y)$ is a Metzler matrix, i.e. a matrix whose off diagonal terms are nonnegative \cite{MR94c:34067,0458.93001}. These matrices are also known as  quasi-positive matrix \cite{0821.34003,MR1993355}. This proves that $F$ is monotone \cite{0821.34003,MR2182759}.
Now, we have  to check the strict sublinearity. We use the
equivalent definition of \cite{MR2182759}, using the standard
ordering of $\R^n$ and the  classical notations  $x \leq y $
if, for any index $i$, $x_i \leq y_i$; $x <y$ if $x\leq y$ and
$x \neq y$  ;   $x \ll y$ if  $x_i <y_i$ for any index $i$ ;

\n
 $F$ is strongly sublinear if

\[0<\lambda <1,  \;\; w \gg 0 \Longrightarrow \; \lambda \, F(w) \gg F(\lambda \, w).\]

\n With $x \gg 0$ and $y \gg 0$, since $\M$ is irreducible, we must have $\M\begin{pmatrix}
x\\y \end{pmatrix} \gg0$, and hence  we have $L \,y \gg 0$ and
$M\, x \gg 0$.  Thus,  $\mu_h  + \,\lambda \,  L\, y
\ll \mu_h + \, L\, y$ and  a similar inequality
$\mu_v  + \, \lambda M\, x \ll \mu_v 
+ M\, x $. This proves the strict sublinearity. Using the
formula for the Jacobian of $F$, we have

\[ JF(0
,0) =
\begin{bmatrix}
 0 &  \diag(\mu_h+\gamma_h)^{-1}\diag(\bar{N}_h)\, L \\
 \\
\diag(\mu_v)^{-1}\, \diag( \bar N_v ) \, M & 0
\end{bmatrix}\]

\n This matrix is irreducible, since $\M$ is irreducible,  and  $\rho(JF(0,0))=\mathcal
R_0$. All the requirements of Theorem~\ref{thm:thieme} are satisfied.
This proves that there exists  a unique positive endemic
equilibrium in $\R^n_+$ when $\mathcal R_0 >1$. Moreover,
looking at the expression of $F$, it is clear that this
equilibrium is in the compact $\Omega$.

\m \n We will prove the asymptotic stability of this positive
equilibrium. The proof is adapted from~\cite{MR87c:92046}, using Krasnosel{$'$}ski{\u\i}'s trick~\cite{MR0181881}. The difference is that we have to
vectorize this proof for the infective of  human host and
vectors.
We will show that the  linearized equation has no solution of
the form $X(t)=\exp(z\, t) \, X_0$ with $X_0 \in \C^{3n}$, $z
\in \C$, $\Re z \geq 0$ for  $X_0$  eigenvector and $z $
corresponding eigenvalue of the Jacobian computed at the
endemic equilibrium. Let $X_0=(U,V,W) \in \C^{3n}$ be such an
eigenvector  for the eigenvalue $z$ . Then

\begin{subequations}
 \begin{align}
 z \, U&= -\diag(\mu_h) \,U - \diag(L\,I_v^*) \, U - \diag(S_h^*) \, L \, W
\label{eigen1}\\
   z\,V &=  \diag(L\, I_v^*  ) \,U -  (\mu_h+\gamma_h) \, V + \diag(S_h^*) \, L \, W  \label{eigen2}\\
   z\,W&= \diag(\bar N_v-I_v^*) \, M\, V -\mu_V \, W - \diag(M\,I_h^*) \, W  \label{eigen3}
\end{align}
\end{subequations}

\n Adding the sub-equations (\ref{eigen1})  and(\ref{eigen2})
we obtain the relation

\[\diag(\mu_h+z\one) \, U= -\diag(\mu_h+\gamma_h +z\one) \, V \]

\n Replacing $U$ in (\ref{eigen2})  and (\ref{eigen3}) yields
after some rearrangements

\begin{multline} \label{krasno}
\begin{bmatrix}
\diag\left (\one + z\diag(\mu_h+\gamma_h)^{-1}\one + \diag(z\one+\mu_h+\gamma_h)\diag(z\one+\mu_h)^{-1}\diag(\mu_h+\gamma_h)^{-1}\, L\,I_v^* \right ) \, V \\
 \\
\diag\left (\one+  z\diag(\mu_v)^{-1}\one+
\diag(\mu_v)^{-1}  M\,I_h^*\right ) \, W
\end{bmatrix} = \\
\\
 \begin{bmatrix}
 0  & \diag(\mu_h+\gamma_h)^{-1}\diag(S_h^*) \, L \\
 \\
 \diag(\mu_v)^{-1} \diag(\bar N_v-I_v^*) \, M& 0
\end{bmatrix} \,
\begin{bmatrix}
 V\\
 \\
 W
\end{bmatrix}
\end{multline}

 The matrix

\[ H =  \begin{bmatrix}
 0  & \diag(\mu_h+\gamma_h)^{-1}\diag(S_h^*) \, L \\
 \\
 \diag(\mu_v)^{-1}\diag(\bar N_v-I_v^*) \, M& 0
\end{bmatrix} \]
is a nonnegative irreducible matrix, since its associated graph is isomorphic to $\Gamma(\M)$.  From equations
(\ref{subeqnendem2}) and (\ref{subeqnendem3}),  we have that

\[ H \, \begin{bmatrix}
 I_h^*\\
 I_v^*
\end{bmatrix}  = \begin{bmatrix}
 I_h^*\\
 I_v^*
\end{bmatrix}.\]

\n Note that $\begin{bmatrix}
 I_h^*\\
 I_v^*
\end{bmatrix}$ is the positive Perron-Frobenius vector of $H$.

\n We assume that $\Re z \geq 0$. Let $\eta (z)$ be the
minimum of the real part of the  components of the two vectors

\[ z\diag(\mu_h+\gamma_h)^{-1}\one+ \diag(z\one+\mu_h+\gamma_h)\diag(z\one+\mu_h)^{-1}\diag(\mu_h+\gamma_h)^{-1} L\,I_v^*\]
and
\[ z\diag(\mu_v)^{-1}\one+\diag(\mu_v)^{-1}M\,I_h^* \]
Since $\Re z \geq 0$, $I_v^* \gg 0$, $I_h^* \gg 0$,  the irreducibility of $\M$ implies that  we have $\eta(z)>0$.
Taking the absolute values in (\ref{krasno}) gives
\[[ 1+\eta(z) ] \,
\begin{bmatrix}
 | V | \\
 | W|
\end{bmatrix} \leq H \,   \begin{bmatrix}
 | V | \\
 | W|
\end{bmatrix} \]
Let $r$ the minimum number such that
 \[  \begin{bmatrix}
 | V | \\
 | W|
\end{bmatrix} \leq r \, \begin{bmatrix}
 I_h^*\\
 I_v^*
\end{bmatrix}. \]
We now have
\[[ 1+\eta(z) ] \,
\begin{bmatrix}
 | V | \\
 | W|
\end{bmatrix} \leq H \,   \begin{bmatrix}
 | V | \\
 | W|
\end{bmatrix}  \leq  r \, H \, \begin{bmatrix}
 I_h^*\\
 I_v^*
\end{bmatrix}=  r \, \begin{bmatrix}
 I_h^*\\
 I_v^*
\end{bmatrix}.\]
Since $\eta(z) >0$ if $\Re z\geq 0$, we obtain a
contradiction to the minimality of $r$. Thus $\Re z <0$, which
proves the asymptotic stability at the endemic equilibrium.
\end{proof}

\section{Global Dynamics}

\label{sec:global}

In this section, we discuss a number of  results concerning the global dynamics of system~\eqref{npatch4}. We begin by introducing some notation to allow an easier handling of the vector calculations.

\begin{definition}
The entry-wise product for vectors,  the Hadamard product,  will be denoted by $\circ$. Namely, if $(X_1,\ldots,X_n), (Y_1,\ldots,Y_n)\in\R^n$, then
\[
(X_1,\ldots,X_n)\circ(Y_1,\ldots,Y_n)=(X_1Y_1,\ldots,X_nY_n).
\]
 For a vector $\bdX=(X_1,\ldots,X_n)\in\R^n$ and for $f:I\subset\R\to\R$, we shall write
 \[
 f(\bdX)=(f(X_1),\ldots,f(X_n)).
 \]
  In particular, if $X=(X_1,\ldots,X_n)\gg0$, then $X^{-1}=(X_1^{-1},\ldots,X_n^{-1})$.
\end{definition}

 We collect some useful facts about the manipulation of expression involving Hadamard products in the following Lemma:

\begin{lemma}
If $\bdX_1,\ldots,\bdX_m\in\R^n$ and $M\in M_n(\R)$ then we have
\begin{enumerate}
\item $\bdX_1+\cdots+\bdX_m\geq m\sqrt[m]{\bdX_1\hp\dots\hp \bdX_m};$\\
\item $\bdX_1\circ(M\bdX_2)=\diag(\bdX_1)M\bdX_2=\diag(M\bdX_2)\bdX_1;$
\item if $\bdX_1=\bdX_1(t)$, and if $f$ is differentiable then $\dfrac{\rd}{\rd t}f(\bdX_1)=\dot{\bdX}_1\hp f'(\bdX_1)$.
\end{enumerate}
\end{lemma}

It turns out that it is more convenient to work with system~ \eqref{npatch4} in prevalence form,  so that the susceptible population at the  disease-free equilibrium (DFE) , for both host and vector populations in each group, is  unity. Let
\begin{align*}
&D_h=\diag(\bar N_h),\; D_v=\diag(\bar N_v),\\[2mm]
&(X,Y)=D_h^{-1}(S_h, I_h),\quad Z=D_v^{-1}I_v\\[2mm]
&A=LD_v\text{ and }B=MD_h.
\end{align*}
 introduce 

In this case system~\eqref{npatch4} reads
\begin{equation}
\label{eqn:sys_res}
\left\{
\begin{array}{ccc}
\dot{X}&=&\mu_h\circ(\one-X)-\diag(X)AZ\\[2mm]
\dot{Y}&=&\diag(X)AZ-(\mu_h+\gamma_h)\circ Y\\[2mm]
\dot{Z}&=&\diag(\one-Z)BY-\mu_v\circ Z
\end{array}
\right.
\end{equation}
I suggest to use $X^*, Y^*, \NGO^*,...$ instead of $\bar X, \bar Y, \bNGO, ..$ for all what is related to the endemic equilibrium since bar has already been used and we used stars for the EE in the previous section

With this notation, the DFE is $(\one,0,0)$ and  we shall write the EE as $(\bX,\bY,\bZ)$, with
\[ 
\bX_i=\left(\frac{X^*_i}{\bar{N}_{h,i}}\right),\quad  \bY_i=\left(\dfrac{Y^*_i}{\bar{N}_{h,i}}\right)\quad \textbf{and}\quad \bZ_i=\left(\dfrac{Z^*_i}{\bar{N}_{v,i}}\right).
\]
Notice that, since in the new coordinates we have $\diag(\bar N_h)=\diag(\bar N_v)=\diag(\one)$, the next generation operator is now given by
\[
\NGO=\begin{pmatrix}
0&\diag(\mu_h+\gamma_h)^{-1}A\\
\diag(\mu_v)^{-1}B&0
\end{pmatrix}.
\]
Also, the absorbing set can now be written as
\[
K=\left\{(X,Y,Z)\in\R^{3n} \text{ s.t. } 0\leq X+Y\leq \one,\quad 0\leq Z\leq\one \right\}.
\]

We begin with the stability of the DFE when $\RR_0\leq1$:

\begin{theorem}
\label{thm:gs:dfe}
 Assume that hypothesis  \ref{hyp:1} holds and that $\mathcal R_0\leq 1$. Then the DFE is globally asymptotically stable.  If $\RR_0>1$, then the DFE is unstable.
\end{theorem}
\begin{proof}

Since $\NGO$ is irreducible, let $(\alpha, \beta)$  be a left, positive eigenvector  of $\NGO$, associated to the eigenvalue  $\mathcal R_0$.  Let
 \[
 V=\langle \alpha,Y\rangle + \langle \beta,(\mu_h+\gamma_h)\circ \mu_v^{-1}\circ Z\rangle
 \quad\text{and}\quad
 R=\langle\alpha,\diag(\one-X)AZ\rangle + \langle\beta,(\mu_h+\gamma_h)\circ\mu_v^{-1}\circ\,\diag(Z)BY\rangle.
 \]
Notice that $R\geq0$, and that $R$ vanishes in the set
\[ S_0=\{(X,Y,Z) \in K  : \diag(\one-X)AZ=\diag(Z)BY=0\,,\, Y,Z\not=0\}.
\] 
Computing the derivative of $V$ along the flow, we have:
\begin{align*}
\dot{V}&=\langle \alpha,\dot{Y}\rangle + \langle \beta, (\mu_h+\gamma_h)\circ\mu_v^{-1}\circ\dot{Z}\rangle\\
&= \langle \alpha,\diag{(X)}AZ-\left(\mu_h+\gamma_h\right)\circ Y\rangle + \langle \beta, (\mu_h+\gamma_h)\circ\mu_v^{-1}\circ\left(\diag(\one-Z)BY-\mu_v\circ Z\right)\rangle\\
&= \langle \alpha,AZ-\left(\mu_h+\gamma_h\right)\circ Y\rangle + \langle \beta,(\mu_h+\gamma_h)\circ\mu_v^{-1}\circ\left( BY-\mu_vZ\right)\rangle -R\\
&=\left[\mathcal R_0\langle(\mu_h+\gamma_h)\circ\beta,Z\rangle - \langle(\mu_h+\gamma_h)\circ\alpha,Y\rangle +\mathcal R_0\langle(\mu_h+\gamma_h)\circ\alpha,Y\rangle - \langle(\mu_h+\gamma_h)\circ\beta,Z\rangle\right] -R\\
&=\left(\mathcal R_0-1\right)\left[\langle(\mu_h+\gamma_h)\circ\alpha,Y\rangle +\langle(\mu_h+\gamma_h)\circ\beta,Z\rangle\right]-R\\
&\leq 0,
\end{align*} 
 provided that $\mathcal R_0\leq 1$.

Also, notice that when $\mathcal R_0<1$, we have that $\dot{V}=0$ if,
and only if, $Y=Z=0$. Since the DFE is the unique invariant
compact set in this latter case, LaSalle principle implies
that it is globally asymptotically stable.  If $\mathcal R_0=1$ then we observe that $\dot{V}=0$  holds in  $S_0$, which contains the set $\{(X,Y,Z) | Y=Z=0\}$. Nevertheless, it can then be easily verified from system~\eqref{eqn:sys_res} that the DFE is the only invariant set contained in $S_0$. Thus the result follows once again from LaSalle invariance principle.

If $\RR_0>1$, then if  both $Y$ and $Z$ are sufficient close to zero, we have $\dot{V}(\one,Y,Z)>0$. By continuity, this is also true in a neighbourhood of $(\one,0,0)$, and hence the DFE is unstable.

\end{proof}

Before we can tackle the global stability of the endemic equilibrium, when $\mathcal R_0>1$, we need  some preliminary results. 

\begin{lemma}
\label{lem:tt:pev}
Assume that Hypothesis~\ref{hyp:1} holds, and let
\[
\bNGO=
\begin{pmatrix}
0&\diag(\mu_h+\gamma_h)^{-1}\,\diag(\bX)A\\
\diag(\mu_v)^{-1}\,\diag(\one-\bZ)B&0
\end{pmatrix}.
\]
Then, $\bNGO$ is irreducible,  $\rho(\bNGO)=1$ and $\bNGO$ has a  positive left eigenvector $(\xi,\eta)^t$ to $\rho(\bNGO)$. In addition, let
\[
T=\diag(\mu_v)^{-1}\diag(\mu_h+\gamma_h)^{-1}\diag(\bX)A\diag(\one-\bZ)B.
\]
Then $\rho(T)=1$, and  $T^t\eta=\eta$.
\end{lemma}
\vspace{3mm}
\begin{proof}

Since Hypothesis~\ref{hyp:1} holds, we have that $\NGO$ is irreducible, and hence $\bNGO$ is irreducible. From the equilibrium relationship we also have
\[
\bNGO\begin{pmatrix}
\bY\\\bZ
\end{pmatrix}
=
\begin{pmatrix}
\bY\\\bZ
\end{pmatrix},
\]
and hence we have
\[
\rho(\bNGO)=1.
\]
The  remaining claims  follow from Proposition~\ref{prop:factors_irred}.
\end{proof}
Before giving the next definition, we introduce some terminology. For a given digraph $G$,  we will denote its  set of vertices by $\V(G)$, and the set of edges of $G$ by $\E(G)\subset \V(G)\times\V(G)$.  A $c$-edge colored multidigraph ($c$-ECM for short) is a multi-digraph where the parallel edges must have different colors---and therefore a maximum of $c$ parallel edges are allowed. If $G$ is a $c$-ECM, we will write $\CC(G)$  for its set of colors. Thus each edge of  $G$ can be uniquely described as an ordered triple $(v_1,v_2,c)\in \E(G)\subset\V(G)\times\V(G)\times\CC(G)$.  
\begin{definition}[Transitive Contact Multigraph]
Given a contact network $\Gamma(\M)$,  we define the \textit{transitive contact multigraph} (TCM for short) $\Gamma(\MM)$  as the $n$-ECM  of order $n$, obtained from $\Gamma(\M)$ by taking $\V(\Gamma(\MM))=\{1,\ldots,n\}$ and defining $(i,j,k)\in\E(\Gamma(\MM))$ if $L_{i,k}M_{k,j}\not=0$.
\end{definition}
\vspace{3mm}
\begin{remark}
Notice that if we collapse all the parallel edges, then we obtain a graph isomorphic to $\Gamma(LM)$. In particular, Proposition~\ref{prop:factors_irred} then says that $\Gamma(\MM)$ is strongly connected. 
\end{remark}
\begin{remark}
If $(i,j,k)\in\MM$, then this means that an infected host in group $j$ can be the origin of an infection of a host in group $i$ by infecting a vector of group $k$, which then infects the  host in group $i$. Within the fast travelling interpretation, this means that a infect host that is resident in region $j$ can travel to region $k$, where it infects a vector there. This infected vector will subsequently infect a susceptible host of region $i$ that travels to region $k$. See Figure~\ref{fig:exm} for an example of a host-vector contact network, and the corresponding transitive contact multigraph. 
\end{remark}
\begin{figure}[htbp]
	\subfloat[Host-vector contact network]{\includegraphics[scale=1.25]{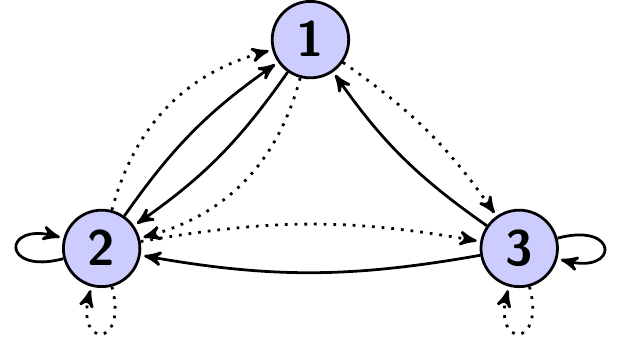}}
	\hfill
	\subfloat[Transitive contact multigraph]{\includegraphics[scale=1.25]{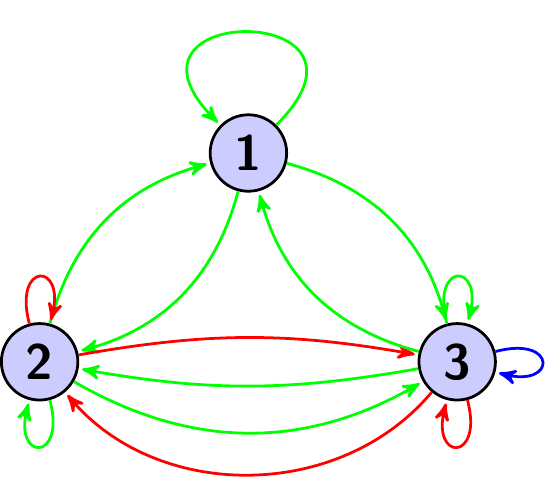}}
	\caption{In (a) we display a host-vector contact network. Within the travelling interpretation of the model, that solid lines indicate the travelling patterns of susceptible hosts (specified by the nonzero entries of $L$), while the dotted lines indicate the travelling pattern of the infected hosts (specified by the nonzero entries of $M$). Notice that, in this example, neither $L$ or $M$ are irreducible, but $\M$ is. In (b) we display the corresponding TCM: the red edges  indicate connections through region 1 (dotted lines in B\&W), the green edges indicate connections through region 2 (dasehd lines in B\&W), and the blue edge indicates a connection trough region 3 (dashed-dotted lines in B\&W).}
	\label{fig:exm}
\end{figure}

We will now give a graph-theoretical interpretation of $\eta$.

\begin{proposition}
Let $\zeta=\diag(\bY)\eta$. Then $\zeta$ spans the kernel of the graph Laplacian of $\MM$. In particular, its entries are given by (a multiple of) the principal minors along the diagonal and, therefore, it is equal to the sum of the weight product of weights of a spanning tree of $\MM$, over all such spanning trees.
\end{proposition}

\begin{proof}
 From the equilibrium relations, we have
\[
T\bY=\bY
\]
and hence
\[
\tT\cdot\one=\one,\text{ where } \tT=\diag^{-1}(\bY)T\diag(\bY).
\]
Thus, we also have
\[
\tT^t\zeta=\zeta,\quad \zeta=\diag(\bY)\eta.
\]
Notice now that
\begin{equation*}
I-\tT^t=
\begin{pmatrix}
1-\tT_{11}&-\tT_{21}&\cdots&-\tT_{n1}\\
-\tT_{12}&1-\tT_{22}&\cdots&-\tT_{n2}\\
\vdots&\vdots&\ddots&\vdots\\
-\tT_{1n}&-\tT_{2n}&\cdots&1-\tT_{nn}
\end{pmatrix}=
\begin{pmatrix}
\sum_{i\not=1}\tT_{i1}&-\tT_{21}&\cdots&-\tT_{n1}\\
-\tT_{12}&\sum_{i\not=2}\tT_{i2}&\cdots&-\tT_{n2}\\
\vdots&\vdots&\ddots&\vdots\\
-\tT_{1n}&-\tT_{2n}&\cdots&\sum_{i\not=n}\tT_{in}
\end{pmatrix},
\end{equation*}
where we have used that
\[
\sum_{i=1}^n\tT_{i,j}=1,\quad j\in\{1,\ldots,n\}.
\]
Therefore $\zeta$ is in the kernel of the matrix Laplacian of $\tT^t$. 

In addition, we have that
\[
\Gamma(\tT^t)=\Gamma(\tT)=\Gamma(T),
\]
and the latter is isomorphic to $\MM$ when collapsing all the parallel edges, and hence the Laplacian of $\MM$ with its edge directions reversed is  $I-\tT^t$. Furthermore, since $\Gamma(M)$ is strongly connected  we have, by Proposition~\ref{prop:factors_irred}, that $\MM$ is also strongly connected and thus the kernel of the associated Laplacian is one-dimensional \cite{Chung:1996}. 
The other claims follow from Kirchhoff's theorem for multigraphs---cf. \cite{Bollobas:1998}.
\end{proof}
\bigskip

\bigskip

\begin{theorem}\label{thm:gs:ee}
 Assume that Hypothesis~\ref{hyp:1} holds and that  $\mathcal R_0>1$. Then the EE is globally asymptotically stable. 
\end{theorem}
\bigskip

\begin{proof}
Let
\[
V=\langle X-\bX\circ\log(X),\eta\rangle + \langle Y-\bY\circ\log(Y),\eta\rangle + \langle Z-\bZ\log(Z),\bxi\rangle,\quad \bxi=(\mu_h+\gamma_h)\circ\mu_v^{-1}\circ\xi,
\]
where $(\xi,\eta)^t$ is the positive left eigenvector of $\bNGO$ as discussed in Lemma~\ref{lem:tt:pev}. In particular, we have that

\[
A^t\diag(\bX)\eta=\mu_v\circ\bxi.
\quad\text{and}\quad
B^t\diag(\one-\bZ)\bxi=(\mu_h+\gamma_h)\circ\eta.
\]
Then
\begin{align*}
\dot{V}&= \langle \dot{X}\circ\left(\one-\bX\circ X^{-1}\right),\eta\rangle + \langle \dot{Y}\circ\left(\one-\bY\circ Y^{-1}\right),\eta\rangle + \langle \dot{Z}\circ\left(\one-\bZ\circ Z^{-1}\right),\bxi\rangle\\
&=\langle \mu_h\circ(\one-X)-\diag(X)AZ-\mu_h\circ(\one-X)\circ\bX\circ X^{-1}+\left(\diag(X)AZ\right)\circ\bX\circ X^{-1},\eta\rangle\\
&\qquad +\langle\diag(X)AZ-(\mu_h+\gamma_h)\circ Y-\left(\diag(X)AZ\right)\circ \bY\circ Y^{-1} +(\mu_h+\gamma_h)\circ\bY,\eta\rangle\\
&\qquad +\langle\diag(\one-Z)BY-\mu_v\circ Z-\left(\diag(\one-Z)BY\right)\circ\bZ\circ Z^{-1}+\mu_v\circ\bZ,\bxi\rangle\\
&=\langle\mu_h\circ\left( \one+\bX-X-\bX\circ X^{-1}\right),\eta\rangle  +\langle(AZ)\circ\bX,\eta\rangle-\langle\mu_v\circ  Z,\bxi\rangle -\langle (\mu_h+\gamma_h)\circ Y,\eta\rangle\\
&\qquad +\langle (\mu_h+\gamma_h)\circ\bY,\eta\rangle - \langle\left(\diag(X)AZ\right)\circ\bY\circ Y^{-1},\eta\rangle + \langle\diag(\one-Z)BY,\bxi\rangle\\
&\qquad -\langle\left(\diag(\one-Z)BY\right)\circ\bZ\circ Z^{-1},\bxi\rangle+\langle\mu_v\circ\bZ,\bxi\rangle.
\end{align*}
Now observe that
\[
\langle(AZ)\circ\bX,\eta\rangle=\langle\diag(\bX)AZ,\eta\rangle=\langle Z,A^t\diag(\bX)\eta\rangle=\langle\mu_v\circ Z,\bxi\rangle.
\]
Also, from the equilibrium equations:
\[
(\mu_h+\gamma_h)\circ\bY=\mu_h\circ(\one-\bX)
\quad\text{and}\quad
\diag(\bX)A\bZ=\mu_h\circ(\one-\bX).
\]
Thus,
\[
\langle\mu_v\circ\bZ,\bxi\rangle=\langle\bZ,A^t\diag(\bX)\eta\rangle=\langle \mu_h\left(\one-\bX\right),\eta\rangle.
\]
Combining all this information, we find that
\begin{align*}
\dot{V}&=\langle \mu_h\circ\left(3\one-\bX-X-\bX\circ X^{-1}\right),\eta\rangle-\langle(\mu_h+\gamma_h)\circ Y,\eta\rangle
-\langle\left(\diag(X)AZ\right)\circ\bY\circ Y^{-1},\eta\rangle\\
&\qquad + \langle\diag(\one-Z)BY,\bxi\rangle - \langle\left(\diag(\one-Z)BY\right)\circ\bZ\circ Z^{-1},\bxi\rangle\\
&= \langle\mu_h\circ\left( 3\one-\bX-X-\bX\circ X^{-1}\right),\eta\rangle  + \langle\diag(\one-\bZ)BY,\bxi\rangle
-\langle(\mu_h+\gamma_h)\circ Y,\eta\rangle\\
&\qquad -\langle\left(\diag(X)AZ\right)\circ\bY\circ Y^{-1},\eta\rangle + \langle\diag(\bZ-Z)BY,\bxi\rangle
-\langle\left(\diag(\one-Z)BY\right)\circ\bZ\circ Z^{-1},\bxi\rangle.
\end{align*}
We also have
\begin{align*}
\langle\diag(\one-\bZ)BY,\bxi\rangle&=\langle Y,B^t\diag(\one-\bZ)\bxi\rangle\\
%(\mu_h+\gamma_h)\langle Y,\dfrac{1}%{\mu_v(\mu_h+\gamma_h)}B^t\diag(\one-\bZ)A^t\diag(\bX)\eta\rangle\\
%&=(\mu_h+\gamma_h)\langle Y,T^t\eta\rangle\\
&=\langle(\mu_h+\gamma_h)\circ Y,\eta\rangle.
\end{align*}
and
\begin{align*}
&\langle\diag(\bZ-Z)BY,\bxi\rangle
-\langle\left(\diag(\one-Z)BY\right)\circ\bZ\circ Z^{-1},\bxi\rangle\\
&\qquad = \left\langle\left[2\bZ-Z-\bZ\circ Z^{-1}\right]\circ BY,\bxi\right\rangle.
\end{align*}
Hence, we are left with
\begin{align*}
\dot{V}&=\langle\mu_h\circ\left( 3\one-\bX-X-\bX\circ X^{-1}\right),\eta\rangle 
+ \left\langle\left[2\bZ-Z-\bZ\circ Z^{-1}\right]\circ BY,\bxi\right\rangle \\
&\qquad - \langle\left(\diag(X)AZ\right)\circ\bY\circ Y^{-1},\eta\rangle .
\end{align*}
Now we write
\[
\one=\bX + \one-\bX
\quad\text{and}\quad
\one=\bZ + \one-\bZ.
\]
Then, we also have
\[
-X-\bX^2\circ X^{-1}\leq -2\bX,
\]
and analogously for $Z-\bZ^2\circ Z^{-1}$. 

Therefore, we find
\begin{align*}
\dot{V}&\leq 3\langle \mu_h\circ\left(\one-\bX\right),\eta\rangle -\langle\mu_h\circ \bX\circ(\one-\bX)\circ X^{-1},\eta\rangle\\
&\qquad -\langle \bZ\circ(\one-\bZ)\circ Z^{-1}\circ(BY),\bxi\rangle
- \langle\left(\diag(X)AZ\right)\circ\bY\circ Y^{-1},\eta\rangle.
\end{align*}
Notice that the inequality above for $\dot{V}$ is strict, except when $X=\bX$ and $Z=\bZ$.

Since
\[
\bxi=\diag(\mu_v)^{-1}A^t\diag(\bX)\eta,
\]
we  can then write
\begin{align*}
\dot{V}&\leq 3\langle \mu_h\circ\left(\one-\bX\right),\eta\rangle -\langle \mu_h\circ\bX\circ(\one-\bX)\circ X^{-1},\eta\rangle\\
&\qquad - \langle \mu_v^{-1}\circ \bX\circ A\left(\bZ\circ(\one-\bZ)\circ Z^{-1}\circ(BY)\right),\eta\rangle
- \langle\left(\diag(X)AZ\right)\circ\bY\circ Y^{-1},\eta\rangle.
\end{align*}
Let
\[
\bA=\diag(\bX)A\diag(\bZ)
\quad\text{and}\quad
\bB=\diag(\mu_v)^{-1}\diag(\bZ)^{-1}\diag(\one-\bZ)B\diag(\bY).
\]
Then
$
\bA\one=\mu_h\circ(\one-\bX)
\quad\text{and}\quad
\bB\one=\one.
$
We can then write
\begin{align*}
\dot{V}&\leq 3\langle \bA\one,\eta\rangle -\langle \left(\bA\one\right)\circ\bX\circ X^{-1},\eta\rangle\\
&\qquad - \langle  \bA\left( \bZ\circ Z^{-1}\circ \left(\bB\left(Y\circ\bY^{-1}\right)\right)\right),\eta\rangle
- \langle X\circ\bX^{-1}\circ\left(\bA \left(Z\circ\bZ^{-1}\right)\right)\circ\bY\circ Y^{-1},\eta\rangle\\
&=\sum_{i=1}^n\eta_i\left[3\left(\bA\one\right)_i-\dfrac{\bX_i}{X_i}\left(\bA\one\right)_i-\left(\bA\left( \bZ\circ Z^{-1}\circ \left(\bB\left(Y\circ\bY^{-1}\right)\right)\right)\right)_i-\dfrac{X_i\bY_i}{\bX_iY_i}\left(\bA \left(Z\circ\bZ^{-1}\right)\right)_i\right]\\
&=\sum_{i,j=1}^n\eta_i\bA_{i,j}\left[3-\dfrac{\bX_i}{X_i}-\dfrac{\bZ_j}{Z_j}\left(\bB\left(Y\circ\bY^{-1}\right)\right)_j-\dfrac{X_i\bY_iZ_j}{\bX_iY_i\bZ_j}\right]\\
&=\sum_{i,j,k=1}^n\eta_i\bA_{i,j}\bB_{j,k}\left[3-\dfrac{\bX_i}{X_i}-\dfrac{\bZ_jY_k}{Z_j\bY_k}-\dfrac{X_i\bY_iZ_j}{\bX_iY_i\bZ_j}\right]\\
&=H_n.
\end{align*}

Before proceeding, we recall that a unicyclic graph is a graph with exactly one cycle \cite{Knuth:1997}. Given  the graph $\MM$, we shall denote by $\DD(n,l)$  the set of unicyclic subgraphs of $\MM$, that has order $n$, with cycle of length $l$. Recalling that $\MM$ is a $n$-ECM, we notice that, in a similar way as in Guo et al \cite{gls:2006,gls:2008} we have
\begin{align*}
H_n&=\sum_{i,j,k}^n\eta_i\bA_{i,j}\bB_{j,k}\left[3-\dfrac{\bX_i}{X_i}-\dfrac{Y_k}{\bY_k}\dfrac{\bZ_j}{Z_j}-\dfrac{X_i\bY_i}{\bX_iY_i}\dfrac{Z_j}{\bZ_j}\right]\\
&=\sum_{l=1}^n\left\{\sum_{Q\in\DD(n,l)}  \left(\prod_{(k,h,j)\in \E(CQ)}\bA_{k,j}\bB_{j,h}\right)\right.\\
&\qquad \left.\times \sum_{(r,m,j)\in \E(CQ)}\left[3-\dfrac{\bX_r}{X_r}-\dfrac{Y_m}{\bY_m}\dfrac{\bZ_j}{Z_j}-\dfrac{X_r\bY_r}{\bX_rY_r}\dfrac{Z_j}{\bZ_j}\right]\right\},\\
\end{align*}
where $CQ$ denotes the unique cycle in the unicyclic graph $Q$. Along such a cycle, we have

\begin{align*}
&\sum_{(r,m,j)\in \E(CQ)}\left[3-\dfrac{\bX_r}{X_r}-\dfrac{Y_m}{\bY_m}\dfrac{\bZ_j}{Z_j}-\dfrac{X_r\bY_r}{\bX_rY_r}\dfrac{Z_j}{\bZ_j}\right]\\
=&3|\E(CQ)|-\sum_{(r,m,j)\in \E(CQ)}\left[\dfrac{\bX_r}{X_r}+\dfrac{Y_m}{\bY_m}\dfrac{\bZ_j}{Z_j}+\dfrac{X_r\bY_r}{\bX_rY_r}\dfrac{Z_j}{\bZ_j}\right]\\
\leq&3|\E(CQ)|-3|\E(CQ)|\left[\prod_{(r,m,j)\in \E(CQ)}\dfrac{Y_m\bY_r}{\bY_mY_r}\right]^{1/3|\E(CQ)|}\\
&=0.&
\end{align*}
Hence, we have that $H_n\leq0$, with equality being attained only when
\[
\dfrac{\bX_r}{X_r}=\dfrac{Y_m}{\bY_m}\dfrac{\bZ_j}{Z_j}=\dfrac{X_r\bY_r}{\bX_rY_r}\dfrac{Z_j}{\bZ_j},\quad (r,m,j)\in \E(CQ).
\]
But since, we have $\dot{V}\leq H_n$, with equality only when $X=\bX$ and $Z=\bZ$, we find that $\dot{V}\leq0$, with equality attained only when, for each $Q\in\DD(n,l)$, $l=1,\ldots,n$, we have
\[
1=\dfrac{Y_m}{\bY_m}=\dfrac{\bY_r}{Y_r},\quad (r,m,\cdot)\in \E(CQ).
\]
But  since $CQ$ is a cycle, we have that 
\[
Y_r=\bY_r,\quad r \in \V(CQ).
\]
Since $\M$ is irreducible, we have that $\bA\bB$ is also irreducible by Proposition~\ref{prop:factors_irred}.  Thus,  we have that any two vertices will be in some unicyclic graph, and hence we have equality only when
\[
Y=\bY.
\]
\end{proof}
\section{Discussion}

\label{sec:concl}

We have considered a class of multi-group models for vector-borne  diseases. This class is a natural extension of the classical Bailey-Dietz model and it is a natural candidate for  modeling the impact of fast urban movement in  some vector transmitted diseases, as for instance, in the case of dengue fever---cf. \cite{Adams,Cosner2009,Alvimetal2013}. The host-vector interaction along the  network gives rise to what we call the host-vector contact network---denoted by $\Gamma(\M)$---and that has a number of distinguishing features from the networks that arise in directly transmitted diseases. The most striking one is, perhaps, that the irreducibility  of the circulation topology is not sufficient to guarantee the irreducibility  of the host-vector topology. In addition, we also  characterize the irreducibility of $\Gamma(\M)$ through the irreducibility of the product sub-networks. With this assumption we are able to provide a complete analysis of the dynamics in the sense the this class of models possesses the so-called sharp $\RR_0$ property, i.e., $\RR_0$ is a threshold parameter with the disease free equilibrium being both locally and globally asymptotically stable when $\RR_0\leq 1$, and being unstable when $\RR_0>1$. In addition,  an interior equilibrium (the endemic equilibrium) that is biologically feasible, i.e. has positive coordinates, if and only if  $\RR_0>1$. Furthermore, when it exists it is globally asymptotically stable. 

From a mathematical point of view, these results extend previous result of directly transmitted diseases to the class considered here.  The global stability of the disease free equilibrium (which has been obtained by \cite{ding:etal:2012} for a special case, and more restricted conditions)   is a very natural  extension of the argument presented in \cite{gls:2006}; see also \cite{Shuai:Driessche:2013} for a very general presentation of this argument.  The existence, uniqueness and local stability of the endemic equilibrium  shows that the corresponding  results of \cite{MR87c:92046} for sub-populations hold for this class of models.  Finally, the global stability proof brings a new ingredient in the graph-theoretic framework introduced in \cite{gls:2006,gls:2008}: the identification of $\Gamma(\M)$ with a multi-graph---that we have termed a transitive contact multi-graph---which is a $c$-edge colored multi-digraph, and which contains all the information of the host-vector contact network encoded on  a different way. The product of the host and vector networks can then be interpreted as a contact matrix for such a graph, and  that  allows us to organize  the calculation of the Lie-derivative of the Lyapunov function within a similar graph-theoretical framework  of \cite{gls:2006,gls:2008}. 

The analysis presented here shows  that, in spite of the complexity of the  models in the considered class, the long-term global dynamics is very simple. This, however, does not imply that the transient dynamics of the model is necessarily simple, and further studies are necessary. As an example of this complexity, we refer to \cite{Alvimetal2013} which provides  examples of situations---included in the class  analyzed here---that  have a  local group  $\RR_0$ less than unity, but a global $\RR_0$ that is greater than unity---and hence bounded to evolve to an endemic state in the long term. While this duality of local versus global $\RR_0$ has been observed in other contexts---see \cite{Mckenzie:etal:2012}---we believe that it should be further studied and understood in the realm of epidemic  models.

\def\cprime{$'$}

\end{document}